\documentclass[a4paper]{article}

\usepackage[english]{babel}
\usepackage{graphicx}
\usepackage{amsmath}
\usepackage{amssymb}
\usepackage{xspace}
\usepackage{array}
\usepackage{longtable}
\usepackage{microtype}

\graphicspath{{./images/}}

\newtheorem{defin}{Definition}
 
\newtheorem{theo}[defin]{Theorem}
 \newenvironment{theorem}{\begin{theo} \sl}{\end{theo}}
\newtheorem{lem}[defin]{Lemma}
 \newenvironment{lemma}{\begin{lem} \sl}{\end{lem}}
\newtheorem{coro}[defin]{Corollary}
 \newenvironment{corollary}{\begin{coro} \sl}{\end{coro}}

\newenvironment{proof}{\emph{Proof.}}{\hfill $\Box$\\}

\newcommand{\etal}{\emph{et~al.}}

\newcommand{\graph}[1]{\ensuremath{\theta_{(4 k + #1)}}-graph\xspace}
\newcommand{\canon}[2]{\ensuremath{T_{#1 #2}}}

\newcommand{\const}{\ensuremath{\boldsymbol{c}}\xspace}

\begin{document}

\title{The Price of Order\footnote{Research supported in part by NSERC and Carleton University's President's 2010 Doctoral Fellowship.}
}

\author{
\addtocounter{footnote}{1}
Prosenjit Bose\footnote{School of Computer Science, Carleton University, 1125 Colonel By Drive, Ottawa, K1S 5B6, Canada, \texttt{jit@scs.carleton.ca}} 
\and
Pat Morin\footnote{School of Computer Science, Carleton University, 1125 Colonel By Drive, Ottawa, K1S 5B6, Canada, \texttt{morin@scs.carleton.ca}}
\and
Andr\'e van Renssen\footnote{JST, ERATO, Kawarabayashi Large Graph Project}~$^{,}$\footnote{National Institute of Informatics, 2-1-2 Hitotsubashi, Chiyoda-ku, Tokyo, 101-8430, Japan,
\texttt{andre@nii.ac.jp}
}}

\date{}

\maketitle

\begin{abstract}
We present tight bounds on the spanning ratio of a large family of ordered $\theta$-graphs. A $\theta$-graph partitions the plane around each vertex into $m$ disjoint cones, each having aperture $\theta = 2 \pi/m$. An ordered $\theta$-graph is constructed by inserting the vertices one by one and connecting each vertex to the closest previously-inserted vertex in each cone. We show that for any integer $k \geq 1$, ordered $\theta$-graphs with $4k + 4$ cones have a tight spanning ratio of $1 + 2 \sin(\theta/2) / (\cos(\theta/2) - \sin(\theta/2))$. We also show that for any integer $k \geq 2$, ordered $\theta$-graphs with $4k + 2$ cones have a tight spanning ratio of $1 / (1 - 2 \sin(\theta/2))$. We provide lower bounds for ordered $\theta$-graphs with $4k + 3$ and $4k + 5$ cones. For ordered $\theta$-graphs with $4k + 2$ and $4k + 5$ cones these lower bounds are strictly greater than the worst case spanning ratios of their unordered counterparts. These are the first results showing that ordered $\theta$-graphs have worse spanning ratios than unordered $\theta$-graphs. Finally, we show that, unlike their unordered counterparts, the ordered $\theta$-graphs with 4, 5, and 6 cones are not spanners. 
\end{abstract}

\section{Introduction}
In a weighted graph $G$, let the distance $\delta_G(u, v)$ between two vertices $u$ and $v$ be the length of the shortest path between $u$ and $v$ in $G$. A subgraph $H$ of $G$ is a \emph{$t$-spanner} of $G$ if for all pairs of vertices $u$ and $v$, $\delta_H(u, v) \leq t \cdot \delta_G(u, v)$, where $t$ is a real constant and at least $1$. The \emph{spanning ratio} of $H$ is the smallest $t$ for which $H$ is a $t$-spanner. The graph $G$ is referred to as the {\em underlying graph}~\cite{NS-GSN-06}. We consider the situation where the underlying graph $G$ is a straightline embedding of the complete graph on a set of $n$ vertices in the plane. The weight of each edge $u v$ is the Euclidean distance $|u v|$ between $u$ and $v$. A spanner of such a graph is called a \emph{geometric spanner}. We look at a specific type of geometric spanner: $\theta$-graphs. 

Introduced independently by Clarkson~\cite{C87} and Keil~\cite{K88}, $\theta$-graphs are constructed as follows: for each vertex $u$, we partition the plane into $m_c$ disjoint cones with apex $u$, each having aperture $\theta = 2 \pi/m_c$. The $\theta$-graph is constructed by, for each cone with apex $u$, connecting $u$ to the vertex $v$ whose projection along the bisector of the cone is closest. When $m_c$ cones are used, we denote the resulting $\theta$-graph as $\theta_{m_c}$. Ruppert and Seidel~\cite{RS91} showed that the spanning ratio of these graphs is at most $1/(1 - 2 \sin (\theta/2))$, when $\theta < \pi/3$, i.e. there are at least seven cones. 

In this paper, we look at the ordered variant of $\theta$-graphs. The ordered $\theta$-graph is constructed by inserting the vertices one by one and connecting each vertex to the closest previously-inserted vertex in each cone (a more precise definition follows in the next section). These graphs were introduced by Bose~\etal~\cite{BGM04} in order to construct spanners with nice additional properties, such as logarithmic maximum degree and logarithmic diameter. The current upper bound on the spanning ratio of these graphs is $1/(1 - 2 \sin (\theta/2))$, when $\theta < \pi/3$. 

In 2010, Bonichon~\etal~\cite{BGHI10} showed that the unordered $\theta_6$-graph has spanning ratio 2. This was done by dividing the cones into two sets, positive and negative cones, such that each positive cone is adjacent to two negative cones and vice versa. It was shown that when edges are added only in the positive cones, in which case the graph is called the half-$\theta_6$-graph, the resulting graph is equivalent to the TD-Delaunay triangulation (the Delaunay triangulation where the empty region is an equilateral triangle) whose spanning ratio is~2, as shown by Chew~\cite{C89}. An alternative, inductive proof of the spanning ratio of the half-$\theta_6$-graph was presented by Bose~\etal~\cite{BFRV14}. This inductive proof was generalized to show that the \graph{2} has spanning ratio $1 + 2 \sin(\theta/2)$, where $k$ is an integer and at least 1. This spanning ratio is tight, i.e. there is a matching lower bound. Recently, the upper bound on the spanning ratio of the \graph{4} was improved~\cite{BCMRV15} to $1 + 2 \sin(\theta/2) / (\cos(\theta/2) - \sin(\theta/2))$ and those of the \graph{3} and the \graph{5} were improved to $\cos (\theta/4) / (\cos (\theta/2) - \sin (3\theta/4))$.

By applying techniques similar to the ones used to improve the spanning ratio of unordered $\theta$-graphs, we improve the spanning ratio of the ordered \graph{4} to $1 + 2 \sin(\theta/2) / (\cos(\theta/2) - \sin(\theta/2))$ and show that this spanning ratio is tight. Unfortunately, this inductive proof cannot be applied to ordered $\theta$-graphs with an odd number of cones, as the triangle we apply induction on can become larger, depending on the order in which the vertices are inserted. We also show that the ordered \graph{2} ($k \geq 2$) has a tight spanning ratio of $1 / (1 - 2 \sin(\theta/2))$. 

\begin{table}[ht]
  \begin{center}
    \begin{tabular}{| >{\centering\arraybackslash}m{\dimexpr.3\linewidth-2\tabcolsep} || >{\centering\arraybackslash}m{\dimexpr.37\linewidth-2\tabcolsep} | >{\centering\arraybackslash}m{\dimexpr.3\linewidth-2\tabcolsep} |}
    \hline
    & Upper Bound & Lower Bound \\ 
    \hline \hline
    $\theta_3$, $\theta_4$, $\theta_5$, and $\theta_6$-graph & - & \vspace{1ex} Not constant spanners. \\ [2ex]
    \hline
    \graph{2} & $\frac{1}{1 - 2 \sin \left( \frac{\theta}{2} \right)}$, for $k \geq 2$ & \vspace{1ex} $\frac{1}{1 - 2 \sin \left( \frac{\theta}{2} \right)}$ \\ [2ex]
    \hline
    \graph{3} & $\frac{1}{1 - 2 \sin \left( \frac{\theta}{2} \right)}$, for $k \geq 1$  & \vspace{1ex} $\frac{\cos \left( \frac{\theta}{4} \right) + \sin \theta}{\cos \left( \frac{3\theta}{4} \right)}$ \\ [2ex] 
    \hline
    \graph{4} & $1 + \frac{2 \sin \left( \frac{\theta}{2} \right)}{\cos \left( \frac{\theta}{2} \right) - \sin \left( \frac{\theta}{2} \right)}$, for $k \geq 1$  & \vspace{1ex} $1 + \frac{2 \sin \left( \frac{\theta}{2} \right)}{\cos \left( \frac{\theta}{2} \right) - \sin \left( \frac{\theta}{2} \right)}$ \\ [2ex]
    \hline
    \graph{5} & $\frac{1}{1 - 2 \sin \left( \frac{\theta}{2} \right)}$, for $k \geq 1$  & \vspace{1ex} $1 + \frac{2 \sin \left(\frac{\theta}{2}\right) \cdot \cos \left(\frac{\theta}{4}\right)}{\cos \left(\frac{\theta}{2}\right) - \sin \left(\frac{3\theta}{4}\right)}$\\ [2ex] 
    \hline
    \end{tabular}
  \end{center} 
  \caption{An overview of upper and lower bounds on the spanning ratio of ordered $\theta$-graphs. Upper bounds on ordered $\theta$-graphs with $4k+2$, $4k+3$, or $4k+5$ cones by Bose~\etal$^6$.}
  \label{tab:Summary}
\end{table}

Next, we provide lower bounds for ordered $\theta$-graphs with $4k + 3$ and $4k + 5$ cones (see Table~\ref{tab:Summary}). For ordered $\theta$-graphs with $4k + 2$ and $4k + 5$ cones these lower bounds are strictly greater than the worst case spanning ratios of their unordered counterparts. Finally, we show that ordered $\theta$-graphs with 3, 4, 5, and 6 cones are not spanners. For the ordered $\theta_3$-graph this is not surprising, as its unordered counterpart is connected~\cite{ABBBKRTV13}, but not a spanner~{El09}. For the ordered $\theta_4$, $\theta_5$, and $\theta_6$-graph, however, this is a bit surprising since their unordered counterparts have recently been shown to be spanners~\cite{BBCRV2013,BGHI10,BMRV15}. In other words, we show, for the first time, that obtaining the nice additional properties of the ordered $\theta$-graphs comes at a price.

\section{Preliminaries}
We define a \emph{cone} $C$ to be a region in the plane between two rays originating from a vertex referred to as the apex of the cone. When constructing an (ordered) $\theta_{m_c}$-graph, for each vertex $u$ consider the rays originating from $u$ with the angle between consecutive rays being $\theta = 2 \pi / m_c$. Each pair of consecutive rays defines a cone. The cones are oriented such that the bisector of some cone coincides with the vertical halfline through $u$ that lies above $u$. We call this cone $C_0$ of $u$ and we number the cones in clockwise order around $u$ (see Figure~\ref{fig:Cones}). The cones around the other vertices have the same orientation as the ones around $u$. We write $C_i^u$ to indicate the $i$-th cone of a vertex $u$. For ease of exposition, we only consider point sets that satisfy the following \emph{general position} assumption: no two vertices lie on a line parallel to one of the rays that defines a cone boundary and no two vertices lie on a line perpendicular to the bisector of a cone. 

\begin{figure}[ht]
  \center
  \begin{minipage}[t]{0.45\textwidth}
    \begin{center}
      \includegraphics{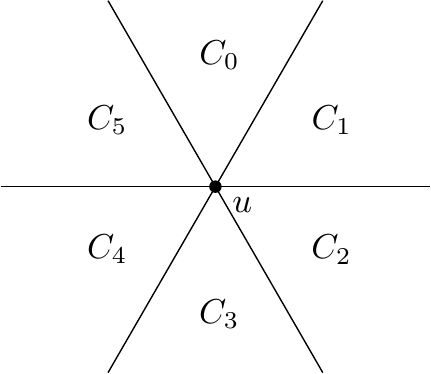}
    \end{center}
    \caption{The cones having apex $u$ in the (ordered) $\theta_6$-graph.\newline}
    \label{fig:Cones}
  \end{minipage}
  \hspace{0.05\linewidth}
  \begin{minipage}[t]{0.46\textwidth}
    \begin{center}
      \includegraphics{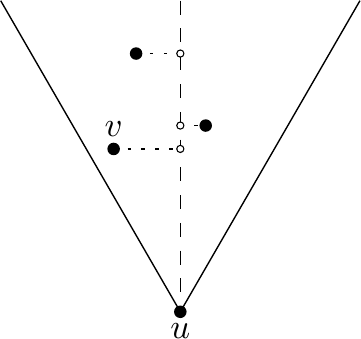}
    \end{center}
    \caption{Three previously-inserted vertices are projected onto the bisector of a cone of $u$. Vertex $v$ is the closest vertex.}
    \label{fig:Projection}
  \end{minipage}
\end{figure}

Given some ordering of the vertices, the ordered $\theta_{m_c}$-graph is constructed as follows: we insert the vertices in the order given by the ordering. When a vertex $u$ is inserted, for each cone $C_i$ of $u$, we add an edge from~$u$ to the closest previously-inserted vertex in that cone, where distance is measured along the bisector of the cone (see Figure~\ref{fig:Projection}). Note that our general position assumption implies that each vertex has a unique closest vertex in each of its non-empty cones and adds at most one edge per cone to the graph. As the ordered $\theta$-graph depends on the ordering of the vertices, different orderings can produce different $\theta$-graphs. 

Given a vertex $w$ in cone $C_i$ of vertex $u$, we define the \emph{canonical triangle} \canon{u}{w} as the triangle defined by the borders of $C_i$ and the line through $w$ perpendicular to the bisector of $C_i$. We use $m$ to denote the midpoint of the side of \canon{u}{w} opposite $u$ and $\alpha$ to denote the smaller unsigned angle between $u w$ and $u m$ (see Figure~\ref{fig:CanonicalTriangle}). Note that for any pair of vertices $u$ and $w$, there exist two canonical triangles: \canon{u}{w} and \canon{w}{u}.

\begin{figure}[h]
  \begin{center}
    \includegraphics{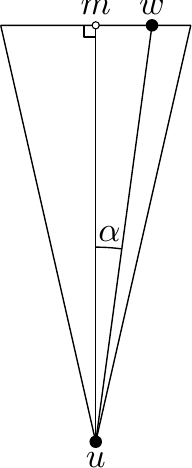}
  \end{center}
  \caption{The canonical triangle \canon{u}{w}.}
  \label{fig:CanonicalTriangle}
\end{figure}

Before we bound the spanning ratios of ordered $\theta$-graphs, we first introduce a few useful geometric lemmas. Note that Lemmas~\ref{lem:ApplyFourPoints} and \ref{lem:CalculationCase} are proven by Bose~\etal~\cite{BCMRV15}. We use $\angle xyz$ to denote the smaller angle between line segments $xy$ and $yz$. 

\begin{lemma}
  \label{lem:ApplyFourPoints} 
  \textbf{\em(Bose~\etal~\cite{BCMRV15}, Lemma~3)} Let $u$, $v$ and $w$ be three vertices in the \graph{x}, where $x \in \{2, 3, 4, 5\}$, such that $w \in C_0^u$ and $v \in \canon{u}{w}$, to the left of $w$. Let $a$ be the intersection of the side of $\canon{u}{w}$ opposite to $u$ with the left boundary of $C_0^v$. Let $C_i^v$ denote the cone of $v$ that contains $w$ and let $c$ and $d$ be the upper and lower corner of $\canon{v}{w}$. If $1 \leq i \leq k-1$, or $i = k$ and $|c w| \leq |d w|$, then $\max \left\{|v c| + |c w|, |v d| + |d w|\right\} \leq |v a| + |a w|$ and $\max \left\{|c w|, |d w|\right\} \leq |a w|$.
\end{lemma}

\begin{figure}[ht]
  \center
  \begin{minipage}[b]{0.45\textwidth}
    \begin{center}
      \includegraphics{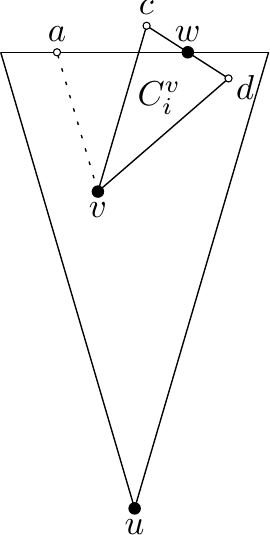}
    \end{center}
    \caption{The situation where we apply Lemma~\ref{lem:ApplyFourPoints}.}
    \label{fig:ApplyFourPoints}
  \end{minipage}
  \hspace{0.05\linewidth}
  \begin{minipage}[b]{0.45\textwidth}
    \begin{center}
      \includegraphics{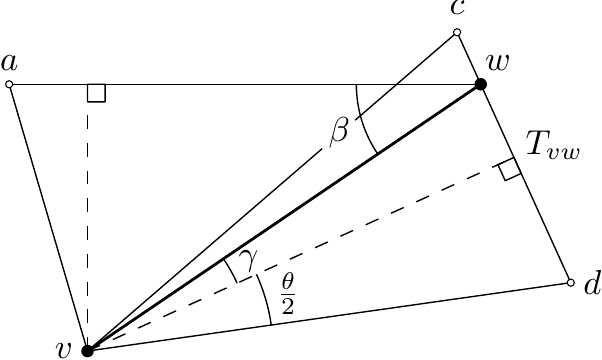}
    \end{center}
    \caption{The situation where we apply Lemma~\ref{lem:CalculationCase}.}
    \label{fig:CalculationLemma}
  \end{minipage}
\end{figure}

\begin{lemma}
  \label{lem:CalculationCase}
  \textbf{\em(Bose~\etal~\cite{BCMRV15}, Lemma~4)} Let $u$, $v$ and $w$ be three vertices in the \graph{x}, where $x \in \{2, 3, 4, 5\}$, such that $w \in C_0^u$, $v \in \canon{u}{w}$ to the left of $w$, and $w \not \in C_0^v$. Let $a$ be the intersection of the side of $\canon{u}{w}$ opposite to $u$ with the left boundary of $C_0^v$. Let $c$ and $d$ be the corners of $\canon{v}{w}$ opposite to $v$. Let $\beta = \angle a w v$ and let $\gamma$ be the unsigned angle between $v w$ and the bisector of \canon{v}{w}. Let \const be a positive constant. If \[
  \const \geq \frac{\cos \gamma - \sin \beta}{\cos \left( \frac{\theta}{2} - \beta \right) - \sin \left( \frac{\theta}{2} + \gamma \right)},\]
  then 
  \[ \max \left\{|v c| + \const \cdot |c w|, |v d| + \const \cdot |d w|\right\} \leq |v a| + \const \cdot |a w|.\]
\end{lemma}

\section{The Ordered $\boldsymbol{\theta_{(4 k + 4)}}$-Graph}
\label{sec:Ordered4k+4}
In this section, we give tight bounds on the spanning ratio of the ordered \graph{4}, for any integer $k \geq 1$. We start by improving the upper bounds. 

\begin{theorem}
  \label{theo:PathLengthOrdered4k+4}
  Let $u$ and $w$ be two vertices in the plane such that $w$ was inserted before $u$. Let $m$ be the midpoint of the side of \canon{u}{w} opposite $u$ and let $\alpha$ be the unsigned angle between $u w$ and $u m$. There exists a path connecting $u$ and $w$ in the ordered \graph{4} of length at most 
  \[\left( \frac{\cos \alpha}{\cos \left(\frac{\theta}{2}\right)} + \const \cdot \left(\cos \alpha \cdot \tan \left(\frac{\theta}{2}\right) + \sin \alpha\right) \right) \cdot |u w|,\] where \const equals $1 / (\cos (\theta/2) - \sin (\theta/2))$.
\end{theorem}
\begin{proof}
  We assume without loss of generality that $w \in C_0^u$ and $w$ lies on or to the right of the bisector of $C_0^u$. We prove the theorem by induction on the rank, when ordered by area (ties are broken arbitrarily), of the canonical triangles \canon{x}{y} for all pairs of vertices where $y$ was inserted before $x$. Let $a$ and $b$ be the upper left and right corners of \canon{u}{w}. Our inductive hypothesis is $\delta(u, w) \leq \max\{|u a| + \const \cdot |a w|, |u b| + \const \cdot |b w|\}$, where $\delta(u,w)$ denotes the length of the shortest path from $u$ to $w$ in the ordered \graph{4} and \const equals $1 / (\cos (\theta/2) - \sin (\theta/2))$. 

  We first show that this induction hypothesis implies the theorem. Basic trigonometry gives us the following equalities: $|u m| = |u w| \cdot \cos \alpha$, $|m w| = |u w| \cdot \sin \alpha$, $|a m| = |b m| = |u w| \cdot \cos \alpha \cdot \tan (\theta/2)$, and $|u a| = |u b| = |u w| \cdot \cos \alpha / \cos (\theta/2)$. Thus the induction hypothesis gives that $\delta(u, w)$ is at most $|u w| \cdot (\cos \alpha / \cos (\theta/2) + \const \cdot (\cos \alpha \cdot \tan (\theta/2) + \sin \alpha))$. 

  \textbf{Base case:} \canon{u}{w} has rank 1. Since this triangle is a smallest triangle where $w$ was inserted before $u$, it is empty: if it is not empty, let $x$ be a vertex in \canon{u}{w}. Since \canon{u}{x} and \canon{x}{u} are both smaller than \canon{u}{w}, the existence of $x$ contradicts that \canon{u}{w} is the smallest triangle where $w$ was inserted before $u$. Since \canon{u}{w} is empty, $w$ is the closest vertex to $u$ in $C_0^u$. Hence, since $w$ was inserted before $u$, $u$ adds an edge to $w$ when it is inserted. Therefore, the edge $u w$ is part of the ordered \graph{4}, and $\delta(u, w) = |u w|$. From the triangle inequality and the fact that $\const \geq 1$, we have $|u w| \leq \max\{|u a| + \const \cdot |a w|, |u b| + \const \cdot |b w|\}$, so the induction hypothesis holds.

  \textbf{Induction step:} We assume that the induction hypothesis holds for all pairs of vertices with canonical triangles of rank up to $j$. Let \canon{u}{w} be a canonical triangle of rank $j+1$.

  If $u w$ is an edge in the ordered \graph{4}, the induction hypothesis follows by the same argument as in the base case. If there is no edge between $u$ and $w$, let $v$ be the vertex in \canon{u}{w} that $u$ connected to when it was inserted, let $a'$ and $b'$ be the upper left and right corners of \canon{u}{v}, and let $a''$ be the intersection of the side of $\canon{u}{w}$ opposite $u$ and the left boundary of $C_0^v$ (see Figure~\ref{fig:TriangleCasesOrdered4k+4}). 

  \begin{figure}[ht]
    \begin{center}
      \includegraphics{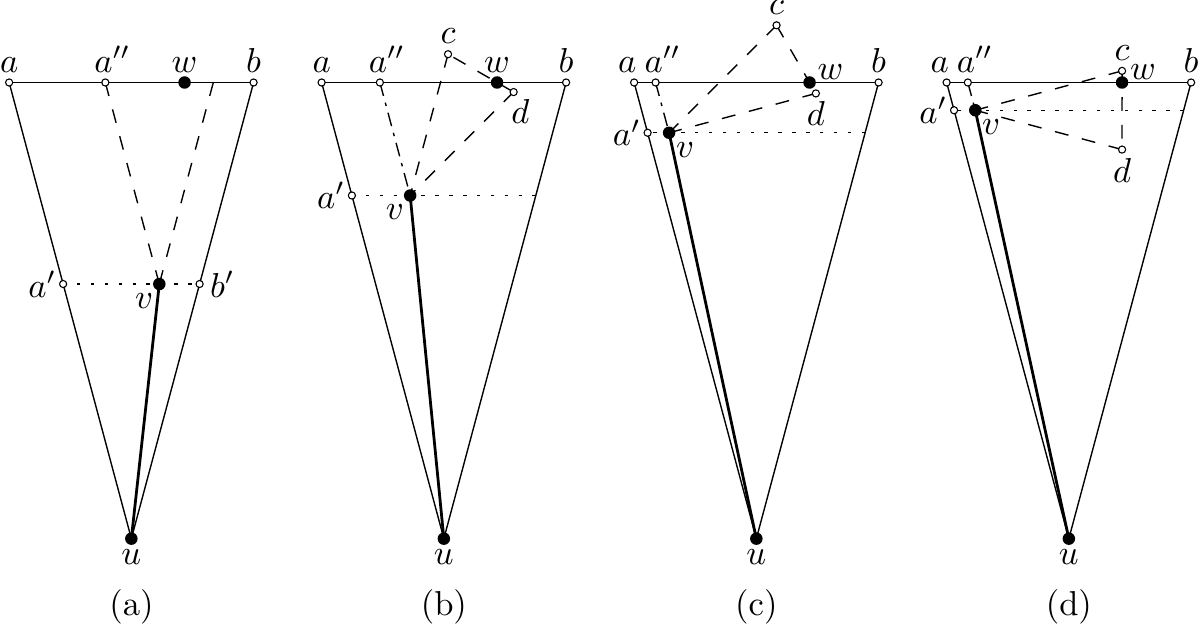}
    \end{center}
    \caption{The four cases based on the cone of $v$ that contains $w$.}
    \label{fig:TriangleCasesOrdered4k+4}
  \end{figure}

  We need to perform case distinction on whether $w$ was inserted before or after $v$, to determine whether we can apply induction on \canon{v}{w} or \canon{w}{v}. Let $c$ and $d$ be the left and right corners of \canon{v}{w} and let $c'$ and $d'$ be the left and right corner of \canon{w}{v}. We note that since the ordered \graph{4} has an even number of cones, $v c w c'$ and $v d w d'$ form two parallelograms. Thus, we have that $|v c| + \const \cdot |c w| = |w c'| + \const \cdot |c' v|$ and $|v d| + \const \cdot |d w| = |w d'| + \const \cdot |d' v|$. Hence, we can assume without loss of generality that the canonical triangle we need to look at is \canon{v}{w}. 

  Without loss of generality, we assume that $v$ lies to the left of or has the same $x$-coordinate as $w$; The case where $v$ lies to the right of $w$ is symmetric to this case. Since we need to show that $\delta(u, w) \leq \max\{|u a| + \const \cdot |a w|, |u b| + \const \cdot |b w|\}$, it suffices to show that $\delta(u, w) \leq |u a| + \const \cdot |a w|$. We perform a case analysis based on the cone of $v$ that contains $w$: (a) $w \in C_0^v$, (b) $w \in C_i^v$ where $1 \leq i \leq k-1$, or $i = k$ and $|c w| \leq |d w|$, (c) $w \in C_k^v$ and $|c w| > |d w|$, (d) $w \in C_{k+1}^v$. To prove that $\delta(u, w) \leq |u a| + \const \cdot |a w|$, it suffices to show that $\delta(v, w) \leq |v a''| + \const \cdot |a'' w|$, as $|u v| \leq |u a'| + \const \cdot |a' v|$ and $v$, $a''$, $a$, and $a'$ form a parallelogram (see Figure~\ref{fig:TriangleCasesOrdered4k+4}).

  \textbf{Case (a):} Vertex $w$ lies in $C_0^v$ (see Figure~\ref{fig:TriangleCasesOrdered4k+4}a). Since \canon{v}{w} has smaller area than \canon{u}{w}, we apply the inductive hypothesis to \canon{v}{w}. Since $v$ lies to the left of or has the same $x$-coordinate as $w$, we have $\delta(v, w) \leq |v a''| + \const \cdot |a'' w|$. 

  \textbf{Case (b):} Vertex $w$ lies in $C_i^v$, where $1 \leq i \leq k-1$, or $i = k$ and $|c w| \leq |d w|$.  Since \canon{v}{w} is smaller than \canon{u}{w}, by induction we have $\delta(v, w) \leq \max\{|v c| + \const \cdot |c w|, |v d| + \const \cdot |d w|\}$ (see Figure~\ref{fig:TriangleCasesOrdered4k+4}b). Since $w \in C_i^v$ where $1 \leq i \leq k-1$, or $i = k$ and $|c w| \leq |d w|$, we can apply Lemma~\ref{lem:ApplyFourPoints}. Note that point $a$ in Lemma~\ref{lem:ApplyFourPoints} corresponds to point $a''$ in this proof. Hence, we get that $\max\left\{|vc| + |cw|, |vd| + |dw|\right\} \leq |va''| + |a''w|$ and $\max\left\{|cw|, |dw|\right\} \leq |a''w|$. Since $\const \geq 1$, this implies that $\max\left\{|vc| + \const \cdot |cw|, |vd| + \const \cdot |dw|\right\} \leq |v a''| + \const \cdot |a'' w|$.

  \textbf{Case (c)} Vertex $w$ lies in $C_k^v$ and $|c w| > |d w|$. Since \canon{v}{w} is smaller than \canon{u}{w} and $|c w| > |d w|$, the induction hypothesis for \canon{v}{w} gives $\delta(v, w) \leq |v c| + \const \cdot |c w|$ (see Figure~\ref{fig:TriangleCasesOrdered4k+4}c). Let $\beta$ be $\angle a'' w v$ and let $\gamma$ be the angle between $v w$ and the bisector of \canon{v}{w}. We note that $\gamma = \theta - \beta$. Hence Lemma~\ref{lem:CalculationCase} gives that $|v c| + \const \cdot |c w| \leq |v a''| + \const \cdot |a'' w|$ holds when $\const \geq (\cos (\theta - \beta) - \sin \beta) / (\cos (\theta/2 - \beta) - \sin (3\theta/2 - \beta))$. As this function is decreasing in $\beta$ for $\theta/2 \leq \beta \leq \theta$, it is maximized when $\beta$ equals $\theta/2$. Hence $\const$ needs to be at least $(\cos (\theta/2) - \sin (\theta/2)) / (1 - \sin \theta)$, which can be rewritten to $1 / (\cos (\theta/2) - \sin (\theta/2))$. 

  \textbf{Case (d)} Vertex $w$ lies in $C_{k+1}^v$ (see Figure~\ref{fig:TriangleCasesOrdered4k+4}d). Since \canon{v}{w} is smaller than \canon{u}{w}, we can apply induction on it. Since $w$ lies above the bisector of \canon{v}{w}, the induction hypothesis for \canon{v}{w} gives $\delta(v, w) \leq |v d| + \const \cdot |d w|$. Let $\beta$ be $\angle a'' w v$ and let $\gamma$ be the angle between $v w$ and the bisector of \canon{v}{w}. We note that $\gamma = \beta$. Hence Lemma~\ref{lem:CalculationCase} gives that $|v d| + \const \cdot |d w| \leq |v a''| + \const \cdot |a'' w|$ holds when $\const \geq (\cos \beta - \sin \beta) / (\cos (\theta/2 - \beta) - \sin (\theta/2 + \beta))$, which is equal to $1 / (\cos (\theta/2) - \sin (\theta/2))$.   
\end{proof}

Since $\cos \alpha / \cos (\theta/2) + (\cos \alpha \cdot \tan (\theta/2) + \sin \alpha) / (\cos (\theta/2) - \sin (\theta/2))$ is increasing for $\alpha \in [0, \theta/2]$, for $\theta \leq \pi/4$, it is maximized when $\alpha = \theta/2$, and we obtain the following corollary: 

\begin{corollary}
  \label{cor:SpanningRatioOrdered4k+4}
  The ordered \graph{4} $(k \geq 1)$ is a $\left( 1 + \frac{2 \sin \left( \frac{\theta}{2} \right)}{\cos \left( \frac{\theta}{2} \right) - \sin \left( \frac{\theta}{2} \right)} \right)$-spanner. 
\end{corollary}

Next, we provide a matching lower bound on the spanning ratio of the ordered \graph{4}. 

\begin{lemma}
  \label{lem:LowerBound4k+4}
  The ordered \graph{4} $(k \geq 1)$ has spanning ratio at least $1 + \frac{2 \sin \left( \frac{\theta}{2} \right)}{\cos \left( \frac{\theta}{2} \right) - \sin \left( \frac{\theta}{2} \right)}$.
\end{lemma}
\begin{proof}
  To prove the lower bound, we first construct a point set, after which we specify the order in which these vertices are inserted into the graph. We fix a vertex $u$ and a vertex $w$ arbitrarily close to the right boundary of $C_0^u$. Next, we place a vertex $v_1$ arbitrarily close to the left corner of \canon{u}{w}, followed by a vertex $v_2$ arbitrarily close to the upper corner of \canon{w}{v_1}. Finally, we repeat the following two steps an arbitrary number of times: we position a vertex $v_i$ arbitrarily close to the left corner of \canon{v_{i-2}}{v_{i-1}}, followed by a vertex $v_{i+1}$ arbitrarily close to the upper corner of \canon{v_{i-1}}{v_i}. Let $v_n$ be the last vertex placed in this fashion. We insert the vertices in the following order: $v_n$, $v_{n-1}$, ..., $v_2$, $v_1$, $w$, $u$. The resulting ordered \graph{4} consists of a single path between $u$ and $w$ and is shown in Figure~\ref{fig:OrderedTheta4k+4}. Note that when a vertex $v$ is inserted, all previously-inserted vertices lie in the same cone of $v$. This ensures that no shortcuts are introduced when inserting $v$. 

  \begin{figure}[ht]
    \begin{center}
      \includegraphics{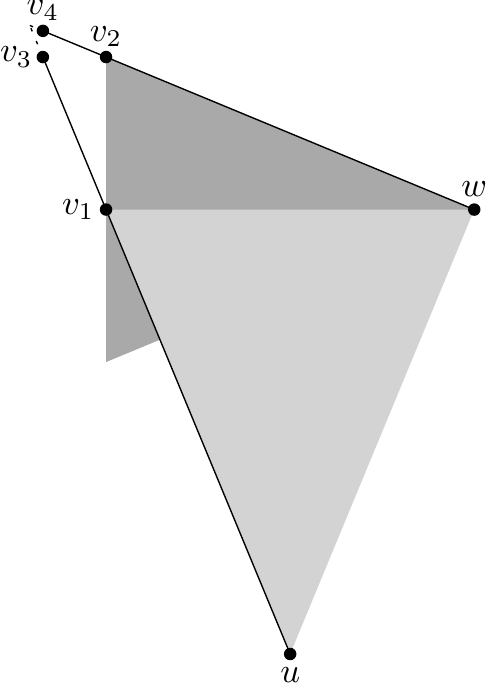}
    \end{center}
    \caption{A lower bound for the ordered \graph{4}.}
    \label{fig:OrderedTheta4k+4}
  \end{figure}

  We note that edges $u v_1$ and edges of the form $v_i v_{i+2}$ (for odd $i \geq 1$) lie on a line. We also note that edges $w v_2$ and edges of the form $v_i v_{i+2}$ (for even $i \geq 2$) lie on a line. Let $x$ be the intersection of these two lines and let $\beta$ be $\angle x w v_1$. Hence, as the number of vertices approaches infinity, the total length of the path approaches $|u x| + |x w|$. Using that $\angle u x w = (\pi - \theta)/2 - \beta$, we compute the following edge lengths: 
  \begin{eqnarray*}
    |u x| &=& |u w|\cdot \frac{\sin \left( \frac{\pi-\theta}{2} + \beta \right)}{\sin \left( \frac{\pi-\theta}{2} - \beta \right)} \\ \\
    |x w| &=& |u w| \cdot \frac{\sin \theta}{\sin \left( \frac{\pi-\theta}{2} - \beta \right)} 
  \end{eqnarray*}

  Since for the ordered \graph{4} $\beta = \theta/2$, the sum of these equalities is $1/(\cos \theta + \tan \theta)$, which can be rewritten to $1 + 2 \sin(\theta/2)/(\cos(\theta/2) - \sin(\theta/2))$.
\end{proof}

\begin{theorem}
  The ordered \graph{4} $(k \geq 1)$ has a tight spanning ratio of $1 + \frac{2 \sin \left( \frac{\theta}{2} \right)}{\cos \left( \frac{\theta}{2} \right) - \sin \left( \frac{\theta}{2} \right)}$. 
\end{theorem}

\section{Lower Bounds}
Next, we provide lower bounds for the ordered \graph{2}, the ordered \graph{3}, and the ordered \graph{5}. For the ordered \graph{2}, this lower bound implies that the current upper bound on the spanning ratio is tight. For the ordered \graph{2} and the ordered \graph{5}, these lower bounds are strictly larger than the upper bound on the worst case spanning ratio of its unordered counterpart. 

\begin{lemma}
  The ordered \graph{2} $(k \geq 2)$ has spanning ratio at least $\frac{1}{1 - 2 \sin(\frac{\theta}{2})}$.
\end{lemma}
\begin{proof}
  To prove the lower bound, we first construct a point set, after which we specify the order in which these vertices are inserted into the graph. We fix a vertex $u$ and a vertex $w$ arbitrarily close to the right boundary of $C_0^u$, and we place a vertex $v_1$ arbitrarily close to the upper left corner of \canon{u}{w}. Next, we repeat the following configuration of six points $l_1$, $v_2$, $l_2$, $r_1$, $v_3$, and $r_2$ an arbitrary number of times: position $l_1$ in \canon{v_1}{w} arbitrarily close to $v_1$; $v_2$ lies in the right corner of \canon{w}{l_1}; $l_2$ is close to the right boundary of \canon{v_1}{v_2} arbitrarily close to $v_1$; $r_1$ lies in the intersection of \canon{v_2}{l_2} and $C_0^{l_2}$ arbitrarily close to $v_2$; $v_3$ is placed in the left corner of the intersection of \canon{l_1}{r_1} and \canon{l_2}{r_1}; vertex $r_2$ is positioned in the intersection of  \canon{v_2}{v_3} and \canon{v_3}{v_2} such that $v_3 r_2$ is parallel to $v_1 w$. Since $v_3 r_2$ is parallel to $v_1 w$, we can repeat placing this configuration, constructing a staircase of vertices (see Figure~\ref{fig:OrderedTheta4k+2}). When we place the $i$-th configuration, we place vertices $l_{2i-1}$, $v_{2i}$, $l_{2i}$, $r_{2i-1}$, $v_{2i+1}$, and $r_{2i}$. 

  \begin{figure}[ht]
    \begin{center}
      \includegraphics{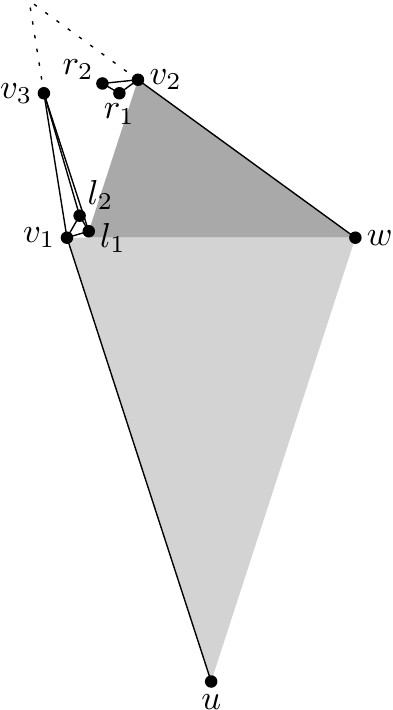}
    \end{center}
    \caption{A lower bound for the ordered \graph{2}.}
    \label{fig:OrderedTheta4k+2}
  \end{figure}

  We insert these vertices into the ordered \graph{2} in the following order: starting from the last configuration down to the first one, insert the vertices of the $i$-th configuration in the order $r_{2i}$, $r_{2i-1}$, $v_{2i+1}$, $l_{2i}$, $l_{2i-1}$, $v_{2i}$. Finally, we insert $w$, $v_1$, and $u$. The resulting ordered \graph{2} is essentially a path between $u$ and $w$ and is shown in Figure~\ref{fig:OrderedTheta4k+2}. 

  We note that edges $u v_1$ and edges of the form $v_i v_{i+2}$ (for odd $i \geq 1$) lie on a line. We also note that edges $w v_2$ and edges of the form $v_i v_{i+2}$ (for even $i \geq 2$) lie on a line. Let $x$ be the intersection of these two lines. Hence, as the number of vertices approaches infinity, the total length of the path approaches $|u x| + |x w|$. Using that $\angle x u w = \theta$, $\angle x w u = (\pi + \theta)/2$, $\angle u x w = (\pi - 3 \theta)/2$, and the law of sines, we compute the following edge lengths: 
  \begin{eqnarray*}
    |u x| &=& |u w| \cdot \frac{\sin \left( \frac{\pi + \theta}{2} \right)}{\sin \left( \frac{\pi - 3 \theta}{2} \right)} \\ \\
    |x w| &=& |u w| \cdot \frac{\sin \theta}{\sin \left( \frac{\pi - 3 \theta}{2} \right)}
  \end{eqnarray*}

  Hence, the spanning ratio of the ordered \graph{2} is at least $(\sin ((\pi + \theta)/2) + \sin \theta)/\sin ((\pi - 3 \theta)/2)$, which can be rewritten to $1/(1 - 2 \sin(\theta/2))$.
\end{proof}

Since Bose~\etal~\cite{BGM04} showed that the \graph{2} has a spanning ratio of at most $1/(1 - 2 \sin(\theta/2))$, this lower bound implies the following theorem. 

\begin{theorem}
  The ordered \graph{2} $(k \geq 2)$ has a tight spanning ratio of $\frac{1}{1 - 2 \sin(\frac{\theta}{2})}$.
\end{theorem}

We also note that since the worst case spanning ratio of the unordered $\theta_{4k+2}$-graph~\cite{BCMRV15} is $1 + 2 \sin(\theta/2)$, this shows that the ordered \graph{2} has a worse worst case spanning ratio. 

\begin{lemma}
  The ordered \graph{3} $(k \geq 1)$ has spanning ratio at least $\frac{\cos \left( \frac{\theta}{4} \right) + \sin \theta}{\cos \left( \frac{3\theta}{4} \right)}$.
\end{lemma}
\begin{proof}
  The proof is analogous to the proof of Lemma~\ref{lem:LowerBound4k+4}, where $\beta = \theta/4$, and shows that the spanning ratio of the ordered \graph{3} is at least $(\sin(\pi/2 - \theta/4) + \sin \theta)/\sin(\pi/2 - 3 \theta/4)$, which can be rewritten to $(\cos(\theta/4) + \sin \theta)/\cos(3\theta/4)$.
\end{proof}

\begin{lemma}
  The ordered \graph{5} $(k \geq 1)$ has spanning ratio at least $1 + \frac{2 \sin \left(\frac{\theta}{2}\right) \cdot \cos \left(\frac{\theta}{4}\right)}{\cos \left(\frac{\theta}{2}\right) - \sin \left(\frac{3\theta}{4}\right)}$.
\end{lemma}
\begin{proof}
  The proof is analogous to the proof of Lemma~\ref{lem:LowerBound4k+4}, where $\beta = 3\theta/4$, and shows that the spanning ratio of the ordered \graph{5} is at least $(\sin(\pi/2 + \theta/4) + \sin \theta)/\sin(\pi/2 - 5 \theta/4)$, which can be rewritten to $1 + 2 \sin(\theta/2) \cdot \cos (\theta/4)/(\cos (\theta/2) - \sin (3\theta/4))$.
\end{proof}

We note that this lower bound on the spanning ratio of the ordered \graph{5} is the same as the current upper bound on $\theta$-routing on the unordered \graph{5}, which is strictly greater than the current upper bound on the spanning ratio of the unordered \graph{5}.

\section{Ordered Theta-Graphs with Few Cones}
In this section we show that ordered $\theta$-graphs with 3, 4, 5, or 6 cones are not spanners. For the ordered $\theta_4$, $\theta_5$, and $\theta_6$-graph, this is surprising, since their unordered counterparts were recently shown to be spanners~\cite{BBCRV2013,BGHI10,BMRV15}. 

For each of these ordered $\theta$-graphs, we build a tower similar to the ones from the previous section. However, unlike the towers in the previous section, the towers of ordered $\theta$-graphs that have at most 6 cones do not converge, thus giving rise to point sets where the spanning ratio depends on the size of these sets. 

\begin{lemma}
  \label{lem:Theta4NotSpanner}
  The ordered $\theta_4$-graph is not a spanner. 
\end{lemma}
\begin{proof}
  To prove that the ordered $\theta_4$-graph is not a spanner, we first construct a point set, after which we specify the order in which these vertices are inserted into the graph. We fix a vertex $u$ and a vertex $w$ slightly to the right of the bisector of $C_0^u$. Next, we place a vertex $v_1$ arbitrarily close to the left corner of \canon{u}{w} and a vertex $v_2$ arbitrarily close to the upper corner of \canon{w}{v_1}. Note that the placement of $v_2$ implies that it lies slightly to the right of the bisector of $C_0^{v_1}$. Because of this, we can repeat placing pairs of vertices in a similar fashion, constructing a staircase of vertices (see Figure~\ref{fig:OrderedTheta4}). Let $v_n$ denote the last vertex that was placed. 

  \begin{figure}[ht]
    \begin{center}
      \includegraphics{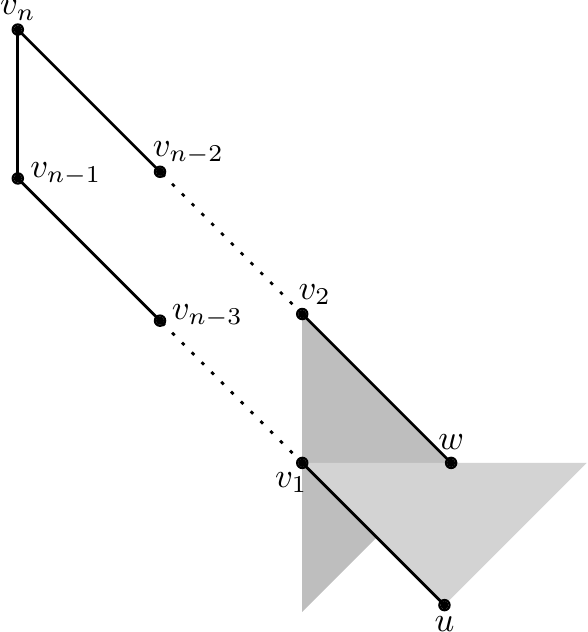}
    \end{center}
    \caption{The ordered $\theta_4$-graph is not a spanner.}
    \label{fig:OrderedTheta4}
  \end{figure}

  We insert these vertices into the ordered $\theta_4$-graph in the following order: $v_n$, $v_{n-1}$, $v_{n-2}$, $v_{n-3}$, ..., $v_2$, $v_1$, $w$, $u$. The resulting ordered $\theta_4$-graph consists of a single path between $u$ and $w$ and is shown in Figure~\ref{fig:OrderedTheta4}. 

  When we take $|u w|$ to be 1, all diagonal edges have length $\const = \sqrt{2}$ and the total length of the path is $1 + n \cdot \sqrt{2}$. Hence, we have a graph whose spanning ratio depends on the number of vertices, implying that there does not exist a constant $t$, such that it is a $t$-spanner.
\end{proof}

\begin{lemma}
  The ordered $\theta_3$-graph is not a spanner. 
\end{lemma}
\begin{proof}
  The proof is analogous to the proof of Lemma~\ref{lem:Theta4NotSpanner}, where $\const = \cos(\pi/6) = \sqrt{3}/2$, and shows that the total length of the path is $1 + n \cdot \sqrt{3}/2$. Hence, we have a graph whose spanning ratio depends on the number of vertices, implying that there does not exist a constant $t$, such that it is a $t$-spanner.
\end{proof}

\begin{lemma}
  The ordered $\theta_5$-graph is not a spanner. 
\end{lemma}
\begin{proof}
  The proof is analogous to the proof of Lemma~\ref{lem:Theta4NotSpanner}, where vertex $w$ is placed such that the angle between $u w$ and the bisector of $C_0^u$ is $\theta/4 = \pi/10$ and $\const = \cos(\pi/10) / \cos(\pi/5)$, and shows that the total length of the path is $1 + n \cdot \cos(\pi/10) / \cos(\pi/5)$. Hence, we have a graph whose spanning ratio depends on the number of vertices, implying that there does not exist a constant $t$, such that it is a $t$-spanner. We note that the placement of $v_i$ (for even $i$) implies that the angle between $v_{i-1} v_i$ and the bisector of $C_0^{v_{i-1}}$ is $\theta/4$. Hence, every pair $v_{i-1}, v_i$ of the staircase has the same relative configuration as the pair $u, w$.
\end{proof}

\begin{lemma}
  The ordered $\theta_6$-graph is not a spanner. 
\end{lemma}
\begin{proof}
  To prove that the ordered $\theta_6$-graph is not a spanner, we first construct a point set, after which we specify the order in which these vertices are inserted into the graph. We fix a vertex $u$, a vertex $w$ arbitrarily close to the right boundary of $C_0^u$, and a vertex $v_1$ arbitrarily close to the left corner of \canon{u}{w}. Next, we repeat the following configuration of four vertices $l_1$, $v_2$, $r_1$, and $v_3$ an arbitrary number of times: position vertex $l_1$ in \canon{v_1}{w} arbitrarily close to $v_1$ (outside \canon{u}{w}); $v_2$ lies in the right corner of \canon{w}{l_1}; place $r_1$ in \canon{v_2}{l_1} arbitrarily close to $v_2$ (outside \canon{w}{l_1}); $v_3$ lies in the left corner of \canon{l_1}{r_1}. Note that the line segment $v_3 r_1$ is parallel to $v_1 w$. Because of this, we can repeat placing four vertices in a similar fashion, constructing a staircase of vertices (see Figure~\ref{fig:OrderedTheta6}). When we place the $i$-th configuration, we place vertices $l_i$, $v_{2i}$, $r_i$, and $v_{2i+1}$. 

  \begin{figure}[ht]
    \begin{center}
      \includegraphics{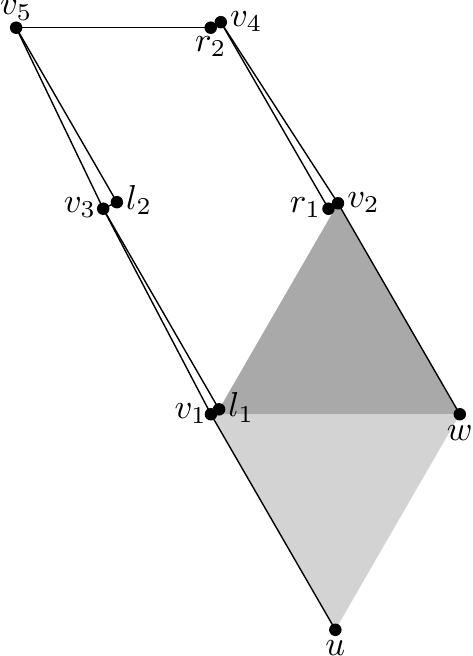}
    \end{center}
    \caption{A lower bound for the ordered $\theta_6$-graph.}
    \label{fig:OrderedTheta6}
  \end{figure} 

  We insert these vertices into the ordered $\theta_6$-graph in the following order: starting from the last configuration down to the first one, insert the vertices of the $i$-th configuration in the order $r_i$, $v_{2i+1}$, $l_i$, $v_{2i}$. Finally, we insert $w$, $v_1$, and $u$. The resulting ordered $\theta_6$-graph is essentially a path between $u$ and $w$ and is shown in Figure~\ref{fig:OrderedTheta6}. 

  When we take $|u w|$ to be 1, we note that every configuration of four vertices extends the path length by 2. Hence, we have a graph whose spanning ratio depends on the number of vertices, implying that there does not exist a constant $t$, such that it is a $t$-spanner.
\end{proof}

\section{Conclusion}
We have provided tight spanning ratios for ordered $\theta$-graphs with $4k + 2$ or $4k + 4$ cones. We also provided lower bounds for ordered $\theta$-graphs with $4k + 3$ or $4k + 5$ cones. The lower bounds for ordered $\theta$-graphs with $4k + 2$ or $4k + 5$ cones are strictly greater than those of their unordered counterparts. Furthermore, we showed that ordered $\theta$-graphs with fewer than 7 cones are not spanners. For the ordered $\theta_4$, $\theta_5$, and $\theta_6$-graph, this is surprising, since their unordered counterparts were show to be spanners~\cite{BBCRV2013,BGHI10,BMRV15}. Thus we have shown for the first time that the nice properties obtained when using ordered $\theta$-graphs come at a price. 

A number of open problems remain with respect to ordered $\theta$-graphs. For starters, though we provided lower bounds for ordered $\theta$-graphs with $4k + 3$ or $4k + 5$ cones, they do not match the current upper bound of $1 / (1 - 2 \sin(\theta/2))$. Hence, the obvious open problem is to find tight matching bounds for these graphs. 

However, more importantly, there is currently no routing algorithm known for ordered $\theta$-graphs. The $\theta$-routing algorithm used for unordered $\theta$-graphs cannot be used, since it assumes the existence of an edge in each non-empty cone. This assumption does not need to hold for ordered $\theta$-graphs, since whether or not an edge is present depends on the order of insertion as well.

\bibliographystyle{abbrv}
\bibliography{references}

\end{document}